\documentclass[sigconf,natbib=true]{acmart}
\AtBeginDocument{%
  }

\setcopyright{none}
\copyrightyear{2018}
\acmYear{2018}
\acmDOI{XXXXXXX.XXXXXXX}
\acmConference[Conference acronym 'XX]{Make sure to enter the correct
  conference title from your rights confirmation emai}{June 03--05,
  2018}{Woodstock, NY}

\acmISBN{978-1-4503-XXXX-X/18/06}

\usepackage{times}
\usepackage{kotex}
\usepackage{soul}
\usepackage{xcolor}
\usepackage{url}
\usepackage{enumitem}
\usepackage{comment}
\usepackage{threeparttable}
\usepackage{lipsum}
\usepackage[linesnumbered,algoruled,boxed,lined,noend]{algorithm2e}
\usepackage{graphicx,subcaption}
\usepackage{fancyhdr} 
\usepackage[utf8]{inputenc}
\usepackage{multirow}
\usepackage[normalem]{ulem}
\usepackage{multicol}
\usepackage{makecell}
\usepackage{cancel}
\usepackage{cellspace}
\usepackage{fontawesome}

\newcommand{\spara}[1]{\smallskip\noindent{\bf #1}}

\newtheorem{definition}{Definition}

\newtheorem{lemma}{Lemma}
\newtheorem{theorem}{Theorem}

\newtheorem{corollary}{Corollary}
\newtheorem{example}{example}

\newcommand{\lse}{\textsf{LSE}}
\newcommand{\Naive}{\textsf{Nai\"ve}}
\newcommand{\Horizontal}{\textsf{LSE$^{H}$}}
\newcommand{\Vertical}{\textsf{LSE$^{HV}$}}
\newcommand{\Diagonal}{\textsf{LSE$^{HVD}$}}
\newcommand{\LSE}{\textsf{LSE}}

\begin{document}

\title{Cohesive Subgraph Discovery in Hypergraphs: A Locality-Driven Indexing Framework}

\author{Song Kim}
\affiliation{%
  \institution{UNIST}
  \country{South Korea}
}
\email{song.kim@unist.ac.kr}

\author{Dahee Kim}
\affiliation{%
  \institution{UNIST}
  \country{South Korea}
}
\email{dahee@unist.ac.kr}

\author{Taejoon Han}
\affiliation{%
  \institution{UNIST}
  \country{South Korea}
}
\email{cheld7132@unist.ac.kr}

\author{Junghoon Kim\textsuperscript{*}}
\affiliation{%
  \institution{UNIST}
  \country{South Korea}
}
\email{junghoon.kim@unist.ac.kr}

\author{Hyun Ji Jeong}
\affiliation{%
  \institution{Kongju National University}
  \country{South Korea}
}
\email{hjjeong@kongju.ac.kr}

\author{Jungeun Kim}
\affiliation{%
  \institution{Inha University}
  \country{South Korea}
}
\email{jekim@inha.ac.kr}

\thanks{\textsuperscript{*}Corresponding author.}

\begin{abstract}
  Hypergraphs, increasingly utilised for modelling complex and diverse relationships in modern networks, gain much attention representing intricate higher-order interactions. Among various challenges, cohesive subgraph discovery is one of the fundamental problems and offers deep insights into these structures, yet the task of selecting appropriate parameters is an open question. To handle that question, we aim to design an efficient indexing structure to retrieve cohesive subgraphs in an online manner. The main idea is to enable the discovery of corresponding structures within a reasonable time without the need for exhaustive graph traversals. This work can facilitate efficient and informed decision-making in diverse applications based on a comprehensive understanding of the entire network landscape. Through extensive experiments on real-world networks, we demonstrate the superiority of our proposed indexing technique.
\end{abstract}

\maketitle


\section{INTRODUCTION}\label{sec:introduction}

Modelling relationships among multiple entities intuitively and effectively is a fundamental challenge in network analysis. Hypergraphs offer a powerful solution by capturing higher-order relationships among groups of entities, transcending the pairwise limitations of traditional graphs. Many real-world networks naturally fit into the hypergraph framework, such as co-authorship networks~\cite{lung2018hypergraph}, co-purchase networks~\cite{xia2021self}, and location-based social networks~\cite{kumar2020hpra}.

Given the expressive power of hypergraphs, hypergraph mining has attracted significant attention. In particular cohesive subhypergraph discoveriy has been actively researched, with various models being proposed. These include the $k$-hypercore~\cite{leng2013m}, $(k,l)$-hypercore~\cite{limnios2021hcore}, $(k,t)$-hypercore~\cite{bu2023hypercore}, nbr-$k$-core~\cite{arafat2023neighborhood}, $(k,d)$-core~\cite{arafat2023neighborhood}, and $(k,g)$-core~\cite{kim2023exploring}. The most recent model, the $(k,g)$-core, defines a maximal subhypergraph where each node has at least $k$ neighbours appearing together in at least $g$ hyperedges. Unlike previous models, which primarily focus on hyperedge cardinality, the $(k,g)$-core captures direct cohesiveness among neighbouring nodes, allowing for the discovery of more densely connected substructures.

Given the benefits, the $(k,g)$-core has broad applicability in various domains: 
(1) \ul{Team formation:} In scenarios of collaborative work or developing marketing strategies, the $(k,g)$-core can identify appropriate groups of participants that are closely connected through multiple shared projects or interactions, helping to form cohesive and effective teams.
(2) \ul{Community Search:} Query-centric community discovery is a central task in social network analysis~\cite{sozio2010community}, it mainly relies on models specifically designed for that purpose. Cohesive subgraph models like the $(k,g)$-core offer a robust framework for identifying tightly-connected substructures within a network~\cite{kim2020densely, akbas2017truss, fang2018effective, liu2021efficient, yao2021efficient}.
(3) \ul{Foundation for other hypergraph mining tasks:} The $(k,g)$-core can support tasks such as identifying densest subhypergraphs~\cite{huang1995clusters}, influence maximisation~\cite{arafat2023neighborhood}, and centrality measures~\cite{busseniers2014general}. Its robust structure enables optimising information spread and quantifying node importance in complex networks.

\begin{figure}[t]
\centering
\includegraphics[width=0.99\linewidth]{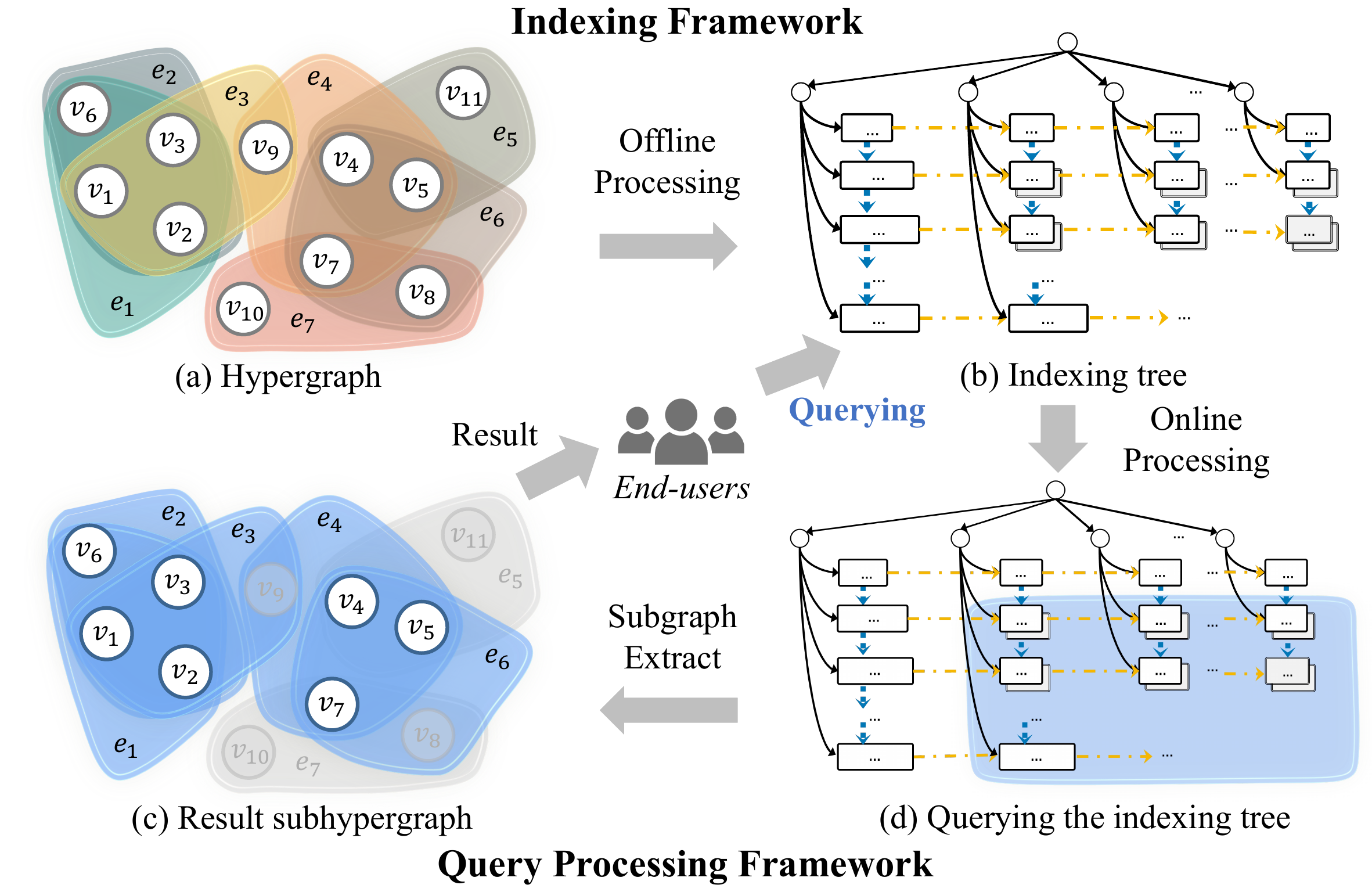}
\vspace{-0.2cm} 
\caption{An illustrative example of proposed framework} 
\vspace{-0.4cm}
\label{fig:graphs}
\end{figure}

Despite its utility, selecting the appropriate $(k,g)$ parameters for a certain task remains challenging. As studied in prior works~\cite{chu2020finding}, choosing suitable parameters for cohesive subgraph models is a difficult open question, especially when varying levels of cohesion are needed. One of the key difficulties arises from the fact that many real-world applications require online query processing~\cite{garcia2008database}, making it impractical to compute cohesive subgraphs from scratch for each query. The ability to rapidly adjust parameters based on user-specified criteria is critical, particularly in domains like social network analysis and recommendation systems, where timely responses are essential. Online processing also enables interactive exploration of hypergraphs allowing users to refine their queries based on intermediate results without waiting for long recomputation steps.

To address these challenges, we propose an indexing-based framework for cohesive subhypergraph discovery. This framework pre-computes an indexing structure over the hypergraph, allowing for efficient query processing by avoiding repetitive subgraph computations. By utilising this structure, we can process online queries for $(k,g)$-core computation, enabling applications in various domains. Our approach dramatically reduces the time complexity associated with recalculating cohesive subgraphs for different parameter values, providing a scalable solution that is well-suited for large-scale hypergraphs. The applications of indexing-based cohesive subgraph discovery include:
(1) \ul{Size-Bounded Cohesive Subgraph Discovery:} In scenarios like size-bounded team formation, conference planning, or social event organisation, it is crucial to assemble groups of participants who are closely connected and of appropriate community size~\cite{yao2021efficient, liu2021efficient}. The index-based approach can efficiently identify candidate groups that meet these criteria.
(2) \ul{User Engagement Analysis:} Social networks are often analysed to understand user engagement at multiple levels. By using the $(k,g)$-core and its index structure, we can evaluate user interactions across various dimensions, helping to identify which users are central to the network cohesion at different levels of connectivity. This insight can be useful for recognising influential users, understanding engagement trends, or even detecting abnormal behaviours like bot activity or coordinated misinformation campaigns.
(3) \ul{Fraud Detection:} In domains such as e-commerce or financial networks, fraudsters often operate within tightly knit groups that exhibit suspiciously high levels of interaction~\cite{pandit2007netprobe, he2022detecting}. The $(k,g)$-core can be applied to identify these densely connected groups of fraudulent actors, based on the frequency and nature of their interactions. By leveraging the indexing-based structure, it becomes feasible to detect such activities in online time, enhancing the ability to prevent fraud before it escalates.

Therefore, in this paper, we propose index-based $(k,g)$-core decomposition techniques, along with efficient online query processing algorithms that leverage these indexing structures. 
The overall framework of our proposed approach is illustrated in Figure~\ref{fig:graphs}. It has two main phases: In the (1) Indexing Framework: the hypergraph is preprocessed into a memory-efficient index structure; and (2) Query Processing Framework: users can perform online queries by dynamically adjusting parameters based on the precomputed index.

\spara{Challenges.} Constructing an efficient index for the $(k,g)$-core is challenging due to the interaction of two variables, $k$ and $g$, which creates significant overlap among nodes with similar coreness values. This results in redundancy and requires careful handling to maintain both space efficiency and query performance.

\begin{itemize}[leftmargin=*]
    \item \ul{Space-efficient indexing structure}: The indexing structure must be designed to be queryable for all possible $(k, g)$ combinations. The hierarchical property of the $(k, g)$-core leads to strong locality between cores with similar values, resulting in significant duplication. Constructing a compact structure to address these challenges is crucial.
    \item \ul{Scalable query processing}: The indexing structure must support fast query processing. The running time of the $(k, g)$-core peeling algorithm increases rapidly with data size, making it inefficient for large networks. Ensuring scalable query processing time regardless of data size is therefore another challenge.
\end{itemize}

\spara{Contributions.} Our contributions are  as follows:
\begin{enumerate}[leftmargin = *]
    \item \ul{Novel problem definition}: To the best of our knowledge, this work is the first to employ an index-based approach for efficiently discovering cohesive subgraphs in hypergraphs.
    \item \ul{Designing new indexing algorithms}: To address the issue of excessive memory usage, we have developed three innovative indexing techniques which utilises the locality characteristics in the core structures of hypergraphs.
    \item \ul{Extensive experiments}: With real-world hypergraph datasets, we conducted extensive experiments, demonstrating the efficiency of our indexing structures and querying algorithms. 
\end{enumerate}

\section{PRELIMINARIES}\label{sec:preliminaries}

A hypergraph is represented as $G=(V, E)$, where $V$ denotes nodes and $E$ denotes hyperedges. In this paper, we consider that $G$ is both undirected and unweighted. For any subset $H\subseteq V$, 
$G[H]=(H, E[H])$ where $E[H]=\{e\cap H | e\in E \wedge e\cap H \neq \emptyset \}$. 
Within hypergraph terminology, it is key to understand the nuanced definitions of terms such as ``degree'', ``neighbours'', and ``cardinality'', especially in contrast to their meanings in traditional graph theory.
\begin{itemize}[leftmargin=*]
    \item Degree: The \textit{degree} of a node is the number of hyperedges that include the node. This differs from a traditional graph, where the degree of a node is the number of edges connected to the node.
    \item Neighbours: The \textit{neighbours} of a node are all other nodes connected to it through shared hyperedges. This contrasts with traditional graphs, where neighbours are only those nodes directly connected to a given node by an edge. 
    \item Cardinality: The \textit{cardinality} of a hyperedge is the number of nodes it contains. This is unique to hypergraphs, as edges in traditional graphs connect only two nodes, with a fixed cardinality of two.

\end{itemize}

\begin{definition}\label{def:kim2023exploring}
    (\underline{$(k,g)$-core}~\cite{kim2023exploring}). Given a hypergraph $G$, $k$, and $g$, $(k,g)$-core is the maximal set of nodes in which each node has at least $k$ neighbours appearing in at least $g$ hyperedges in an induced subhypergraph by the set of nodes.
\end{definition}

In \cite{kim2023exploring}, the uniqueness and hierarchical structure of the $(k,g)$-core are discussed. Specifically, any $(k',g)$-core $\subseteq$ $(k,g)$-core if $k' \geq k$. Similarly, any $(k,g')$-core $\subseteq$ $(k,g)$-core if $g' \geq g$. 

\begin{example}
Consider Figure~\ref{fig:graphs}(a). Nodes ${v_1,\dots,v_7}$ forms the $(2,2)$-core because each node has at least two neighbours, and for each of these neighbours, they share at least two hyperedges. Additionally, nodes ${v_1,\dots,v_5}$ are in $(1,3)$-core since each node has at least one neighbour node with whom it co-occurs in at least three hyperedges, especially $\{v_1,v_2,v_3\}$ appear together $e_1, e_2$, and $e_3$, while nodes $\{v_4,v_5\}$ appear together in hyperedges $e_4, e_5$, and $e_6$.
\end{example}

We next define $k$($g$)-coreness utilised to build indexing structures. 

\begin{definition}
    (\underline{$k$-coreness / $g$-coreness}). Given a hypergraph $G$ and $k$, the $k$-coreness of a node $x$ in $G$, denoted as $c^k(x, G)$, is the maximum $g'$ such that $x$ is in the $(k,g')$-core but not in $(k,g'+1)$-core. Similarly, the $g$-coreness of $x$, denoted as $c^g(x, G)$, is the maximum $k'$ such that $x$ is in the $(k',g)$-core but not in $(k'+1,g)$-core.
\end{definition}

$k$-coreness indicates the maximum number of hyperedges in which a node consistently co-occurs with a fixed number of neighbours., while $g$-coreness reflects how many neighbours a node consistently maintains across varying levels of $g$.


\section{PROBLEM STATEMENT}\label{sec:problemStatement}

While certain properties of the $(k,g)$-core have been identified, selecting appropriate user parameters $k$ and $g$ remains a challenging and unresolved problem, as discussed in Section~\ref{sec:introduction}. Therefore, in this paper, we address this challenge by enabling end-users to adjust these parameters within online time to obtain the desired cohesive subgraph. The research question and our answer are as follows: 

\spara{Research Question.} \ul{Given a hypergraph $G$, is it feasible for a user to dynamically change the $k$ and $g$ values and promptly obtain the corresponding $(k,g)$-core?}

\spara{Answer.} To address this research question, we have developed two distinct approaches, each guided by different strategic considerations: (1) fast query processing time, and (2) space efficiency.

\begin{enumerate}[leftmargin=*]
    \item \textbf{Na\"ive Indexing Approach (\Naive):} This method directly stores the results of the $(k,g)$-core for every possible combination of $k$ and $g$. Thus, end-users can access the desired $(k,g)$-core without any complicated query processing. 
    \item \textbf{Locality-based Space Efficient Approach:} In contrast, this approach focuses on memory efficiency. To achieve this, we introduce several innovative techniques: (1) \ul{Horizontal Locality-based Indexing (\Horizontal):} Utilises the inherent hierarchical structure of $(k,g)$-cores to minimise the amount of stored data; (2) \ul{Vertical Locality-based Indexing (\Vertical):} Employs the hierarchical structure of $(k,g)$-cores across different levels to reduce the duplication of stored nodes; and (3) \ul{Diagonal Locality-based Indexing (\Diagonal):} Focuses on the locality among non-hierarchical cores, especially diagonally adjacent $(k,g)$-cores, by incorporating hidden auxiliary nodes.
\end{enumerate}


\section{\mbox{INDEXING TREE AND QUERY PROCESSING}}\label{sec:indexing_structure}

In this section, we present indexing algorithms for efficiently querying all $(k,g)$-cores to address the research question. We first introduce the Na\"ive Indexing Approach(\Naive), which prioritise query processing performance over memory efficiency. To mitigate its limitations, we then introduce Locality-based Space-Efficient Indexing approaches (\LSE s), which focus on optimising space efficiency while maintaining reasonable query processing times.

\spara{Indexing tree structure.} Figure~\ref{fig:naive} shows the layered design of the proposed indexing structure. The structure is represented as a rooted tree of height $2$ with three layers. The root node (the $g$-layer) corresponds to $g$ value, and the edges from the root node represent different values of $g$. At the $k$-layer, the edges that connect nodes from this layer to the bottom layer represent a specific $k$ value. This structure provides an organised representation of the various combinations of the $(k,g)$ values. The bottom layer consists of $(k,g)$-leaf nodes, which store nodes that are part of the corresponding $(k,g)$-core.
The choice of $g$ or $k$ as the root depends on the application context or user preferences. Notably, it can be easily reversed.

\begin{figure}[t]
\centering
\includegraphics[width=0.95\linewidth]{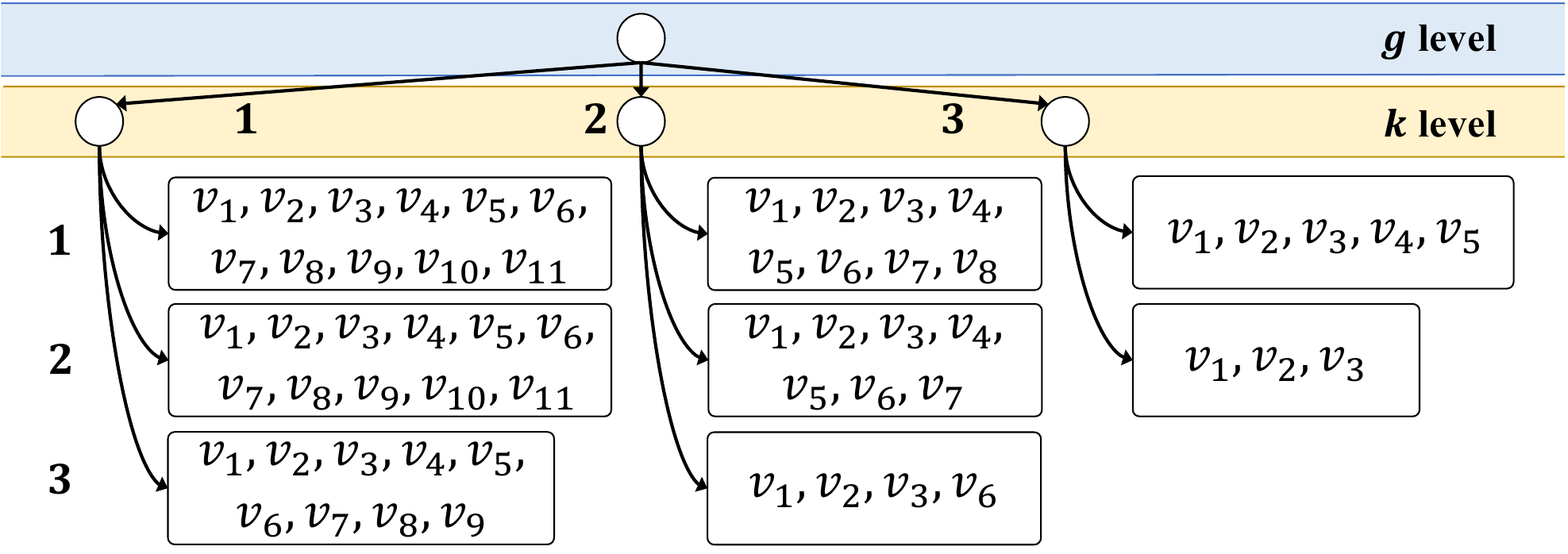}
\vspace{-0.2cm} 
\caption{{\Naive} indexing tree structure} 
\vspace{-0.3cm}
\label{fig:naive}
\end{figure}

\subsection{Na\"ive Indexing Approach (\Naive)}\label{sec:naive}

To immediately retrieve the $(k,g)$-core, we propose a straightforward method called the Naïve Indexing Approach. It constructs an indexing tree that allows for immediate retrieval of any $(k,g)$-core. 

\subsubsection{Indexing framework} 
\textcolor{black!0}{This text is completely invisible.}

\spara{Indexing tree construction.} The construction of the {\Naive} indexing tree involves computing all possible $(k',g')$-cores by increasing $g'$ from $1$ up to the maximum value $g^*$ in the hypergraph. Each computed $(k',g')$-core is then stored in its corresponding leaf node. 

\spara{Leaf node.} Each $(k',g')$-leaf node in the {\Naive} indexing tree contains the set of nodes that belong to the corresponding $(k',g')$-core.

\spara{Query processing.} Query processing of the {\Naive} indexing tree is to find the specific $(k,g)$-leaf, starting from the root node. The identified leaf node contains all the nodes in the $(k,g)$-core.

\subsubsection{Complexity analysis}
\textcolor{black!0}{This text is completely invisible.}

\spara{Construction time complexity.} Let denote the time complexity of the $(k,g)$-core peeling algorithm $O(P)$\footnote{The $(k,g)$-core peeling algorithm forms the basis of our index construction. For simplicity, we represent the complexity of the algorithm as $O(P)$.}, then the construction of the {\Naive} indexing tree takes $O(g^* \cdot P)$, where $g^*$ is the maximum value of $g$ in the hypergraph. 

\spara{Query processing time complexity.} Since the $(k,g)$-leaf node directly contains all the nodes in the $(k,g)$-core, the query result can be retrieved in $O(1)$ time using the {\Naive} indexing tree.

\spara{Space complexity.} In the worst case, each $(k,g)$-core may contain all nodes in $V$. Thus, the total space requires $O(k^*\cdot g^*\cdot |V|)$ space, where $k^*$ is the maximum $k'$ such that the $(k', 1)$-core is non-empty.

\begin{example}
Consider the {\Naive} indexing tree illustrated in Figure~\ref{fig:naive}, which is constructed based on the hypergraph depicted in Figure~\ref{fig:graphs}(a). This tree comprises eight $(k,g)$-leaf nodes, representing a unique combination of $k$ and $g$ values. Furthermore, each $(k, g)$-leaf node contains the nodes associated with the $(k, g)$-core it corresponds to. Within these leaf nodes, certain nodes are highly recurrent. For example, node $v_1$ appears across all $(k,g)$-leaf nodes. 
\end{example}

The Na\"ive indexing technique offers the optimal query processing time, i.e., $O(1)$.  However, it requires $O(k^* \cdot g^*\cdot |V|)$ space, making it inefficient to address huge real-world hypergraphs. To efficiently mitigate the space complexity while preserving the reasonable query processing time, we next present three indexing techniques that mainly leverage the locality properties of the $(k,g)$-core.

\begin{algorithm}[t]
\SetAlgoLined
\small
\SetKwData{break}{break}
\SetKwData{false}{false}
\SetKwData{true}{true}
\SetKwBlock{block}{\FuncSty{/* Lines 5-11 in Algorithm~\ref{alg:enum_function} \quad  */\setcounter{AlgoLine}{11}
}}

\SetKwBlock{blocktwo}{\FuncArgSty{prev} $\leftarrow \emptyset$\;\FuncSty{/* Lines 1-4 in Algorithm~\ref{alg:enum_function} \quad  */\setcounter{AlgoLine}{4}}}

\SetKwFunction{add}{add}

\KwIn{Hypergraph $G=(V,E)$, parameter $g$}
\KwOut{$\{(k,g)\text{-core} \mid k \geq 1 \text{ and } (k,g)\text{-core} \neq \emptyset  \}$}
\For{$k' \leftarrow 1 $ to $|V|$}{
    \If{$|H|\leq k$}{
    \break\;
    }
    \While{\true}{
        changed $\leftarrow \false$\;
        \ForEach{$v \in H$}{
            $N(v) \leftarrow \{(w, c(v, w))| w\in V$, \text{and } $c(v, w) \geq g \}$\;
            \If{$|N(v)| < k'$}{
                $H \leftarrow H \setminus \{v\}$\;
                changed $\leftarrow \true$\;
            }
        }
        \If{changed $=$ \false}{

       \If{$prev \neq \emptyset$}{
            $S.\add(prev \setminus H)$\;
           
           }
           \FuncArgSty{prev} $\leftarrow H$\;
            \break\;
        
        }
    }
}
\If {\FuncArgSty{prev} $\neq \emptyset$}{
$S.\add(prev)$;
}
\Return{$S$}
\caption{\mbox{\FuncSty{enum\_h}: Enum shell structures by fixing $g$}}
\label{alg:enum_h}
\vspace{-0.1cm}
\end{algorithm}

\subsection{\mbox{Horizontal Locality-based Indexing(\Horizontal)}}
In this section, we present Horizontal Locality-based Indexing, which aims to reduce space complexity by utilising hierarchical characteristics of the $(k,g)$-core.

\subsubsection{Indexing framework}
\textcolor{black!0}{This text is completely invisible.}

\spara{Indexing tree construction.} For each $g'$ from $1$ to the maximum value $g^*$, we compute the $g'$-coreness for every node. That is, we calculate $c^{g'}(v, G)$ for each node $v \in V$, where a node $v$ is part of the $(c^{g'}(v, G), g')$-core but not the $(c^{g'}(v, G) + 1, g')$-core. Only the nodes for which $c^{g'}(v, G) = k'$ are stored in the corresponding $(k',g')$-leaf node. As a result, the leaf nodes sharing the same $g'$ value do not contain redundant nodes. The total number of nodes stored under the same $g$-level branch is bounded by $|V|$. Additionally, subsequent leaf nodes, such as $(k',g')$ and $(k' + 1, g')$, are connected by a link. Pseudo description can be checked in Algorithm~\ref{alg:enum_h} and Algorithm~\ref{alg:horizontal_level_index_construction}. 

\spara{Leaf node.} Each leaf node contains both a specific value and a link to the next leaf node in the sequence. The composition of a leaf node in this structure is detailed as follows:
    \begin{itemize}[leftmargin=*]
        \item \textit{Value:} Each $(k', g')$-leaf node in the indexing tree contains a set of nodes, all of which share the same $g'$-coreness, denoted by the value $k'$. This implies that nodes having identical $g'$-coreness for a given $g'$ are grouped in the same leaf node as a value.
        \item \textit{Link:} Given a specific $g'$ value, a leaf node representing a certain ($k'$, $g'$)-coreness is linked to the subsequent leaf node that represents ($k'+1$, $g'$)-coreness. This edge is referred to as the \textit{next link}, and this linked list structure is consistently repeated across every branch corresponding to different $g$ values in the tree.
    \end{itemize}
\spara{Query processing.} To process a query $Q = (k,g)$ with the indexing tree $T$, the following procedure is employed to retrieve the $(k,g)$-core: Initially, we find the $(k,g)$-leaf node. From this leaf node, it iteratively traverses subsequent leaf nodes through \textit{next link} until reaching the terminal leaf node, which is not pointing to a further leaf node. The $(k,g)$-core is then retrieved by aggregating all nodes encountered across these traversed leaf nodes.

\begin{algorithm}[t]
\small
\KwIn{Hypergraph $G=(V,E)$}
\KwOut{Indexing tree $T$}
\SetKwData{break}{break}
\SetKw{Return}{return}
\SetKwData{false}{false}
\SetKwData{continue}{continue}
\SetKwData{prev}{prev}
\SetKwFunction{ins}{insertLeafNode}
\SetKwData{true}{true}
\SetKwFunction{enuma}{enum\_h}
\SetKwFunction{treeInit}{treeInit}
\SetKwBlock{block}{\FuncSty{/* Lines 1-2 in Algorithm~\ref{alg:naive_index_construction} \qquad\quad\   */\setcounter{AlgoLine}{2}    }}

$T \leftarrow$ \treeInit{}, $k' \leftarrow 1$\;
\For{$g' \leftarrow 1$ to $|E|$}{
    $S \leftarrow $ \enuma{$G$, $g'$}\;
    \If{$|S| = \emptyset$}{
        \break\;
    } 
    $prev \leftarrow \emptyset$\;
    \For{$s\in S$}{
         $u \leftarrow$ \ins{$T$, $k'$++, $g'$, $s$}\;
         \If{$prev \neq \emptyset$}{
             $prev$.next $\leftarrow u$\;
         } 
         $prev \leftarrow u$\;
    }
}    
\Return{$T$}
\caption{Horizontal Locality-based indexing}
\label{alg:horizontal_level_index_construction}
\end{algorithm}


\subsubsection{Complexity analysis}
\textcolor{black!0}{This text is completely invisible.}

\spara{Construction time complexity.} The time complexity for constructing the {\Horizontal} indexing tree is $O(g^* \cdot (P + k^* \cdot |V|))$. This includes an additional step compared to the {\Naive} indexing tree, namely, performing the set difference operation between pairs of linked leaf nodes. In the worst case, it requires $O(|V|)$ time for each pair. Note that $O(P)$ is the time complexity for $(k,g)$-core computation.

\spara{Query processing time complexity.} Query processing involves traversing the leaf nodes in the indexing tree $T$, starting from the $(k,g)$-leaf node and continuing until the terminal leaf node is reached. 
Thus, the time complexity for retrieving the $(k,g)$-core is $O(k^* \cdot |V|)$.

\spara{Space complexity.} The {\Horizontal} technique can significantly reduce space complexity by storing each node in the hypergraph only once for each $g$ value, unlike the {\Naive} indexing approach which may store redundant node information. This efficiency is achieved through a horizontal connection of leaf nodes and the removal of redundancies across different $k$ values for the same $g$. Consequently, the space complexity is primarily determined by the number of distinct nodes for each $g$ value, leading to a complexity of $O(g^* \cdot |V|)$.


\begin{figure}[t]
\centering
\includegraphics[width=0.9\linewidth]{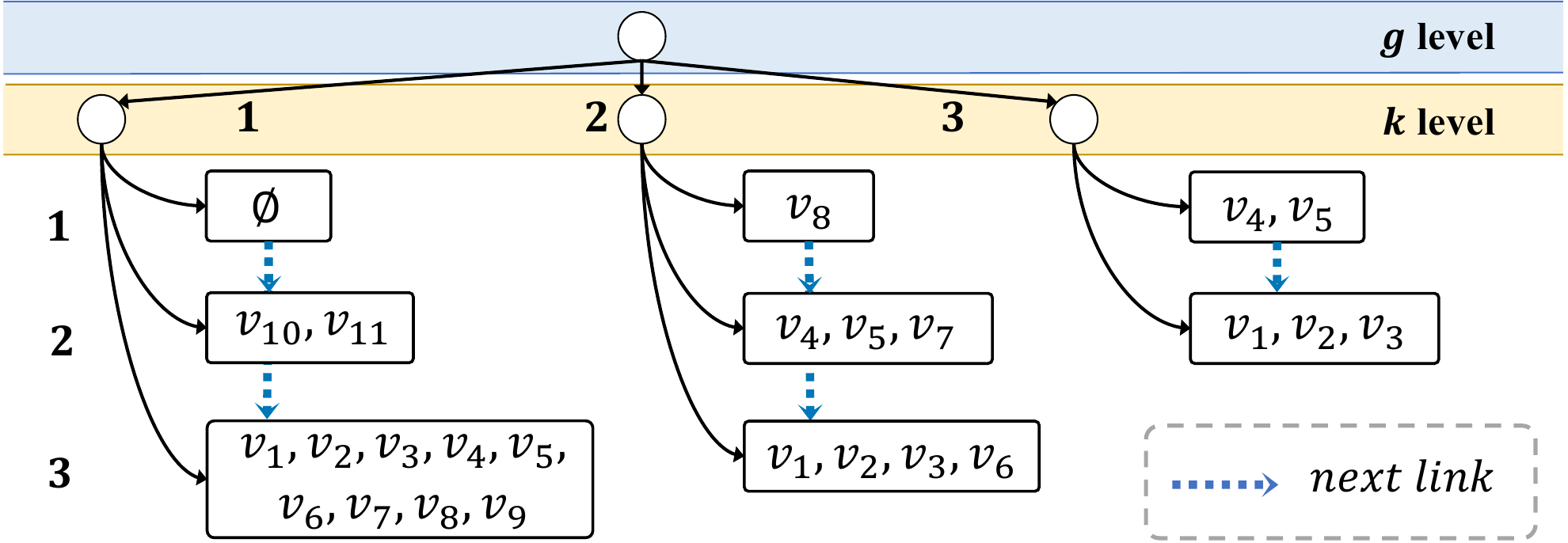}
\vspace{-0.2cm}
\caption{{\Horizontal} indexing tree structure} 
\label{fig:hor_index_example}
\vspace{-0.3cm}
\end{figure}

\begin{example}
Figure~\ref{fig:hor_index_example} presents the {\Horizontal} indexing tree applied to the hypergraph $G$ presented in Figure~\ref{fig:graphs}(a). In this representation, nodes $\{v_1,\dots,v_9\}$ are categorised within the $(3,1)$-leaf. This indicates that these nodes are included in both the $(1,1)$-core and $(2,1)$-core, i.e., $\{v_1,\dots,v_9\} = (3,1)\text{-leaf} \subseteq (2,1)\text{-leaf} \subseteq (1,1)\text{-leaf}$.
When processing a query $Q = (2,1)$, the method involves traversing from the $(2,1)$-leaf node to $(3,1)$-leaf node through the next link, subsequently aggregating the nodes from these leaf nodes to obtain the $(2,1)$-core.
Note that a node $v$ may appear in multiple leaf nodes across different $g$ branches if its $k$-coreness is greater than $1$. For example, nodes $v_1, v_2$, and $v_3$ are present in all $g$ branches.
\end{example}


\subsection{Vertical Locality-based Indexing(\Vertical)}

In this section, we introduce the Vertical Locality-based Indexing technique, developed to improve space complexity by simultaneously considering both $k$ and $g$ parameters. Unlike {\Horizontal}, which focuses on overlaps within the same $g$ value, {\Vertical} addresses overlaps across different $k$ and $g$ values at the same time. For instance, while {\Horizontal} eliminates duplicate nodes between $(1,2)$-core, $(2,2)$-core, and $(3,2)$-core, it does not consider overlaps between cores like $(1,1)$-core and $(2,2)$-core or $(3,3)$-core respectively. {\Vertical} technique, therefore, introduces additional links between leaf nodes of the same $k$ value, preserving efficient query processing by accounting for duplicate relationships across both $k$ and $g$ values. 

\subsubsection{Indexing framework}
\textcolor{black!0}{This text is completely invisible.}

\spara{Indexing tree construction.} {\Vertical} indexing tree construction is based on the {\Horizontal} indexing tree. For each $k'$ value, we handle redundancies between adjacent leaf nodes, such as the $(k',g')$-leaf node and the $(k',g'+1)$-leaf node, connecting them together. The structure is illustrated in Figure~\ref{fig:vertical_tree}. As with the {\Horizontal} method, the total number of distinct nodes in each $g$-level branch is also bounded by $|V|$. While this indexing technique reduces redundancy by considering nodes with identical $k$-coreness and $g$-coreness, it does not address the case where a node has varying $k$-coreness across different branches corresponding to distinct $g$ values within the tree. Pseudo description can be checked in Algorithm~\ref{alg:vertical}. 

\spara{Leaf node.} In the Vertical Locality-based Indexing ({\Vertical}), each leaf node contains a specific value and two distinct links. The components of a leaf node are as follows: 

\begin{itemize}[leftmargin=*]  
    \item \textit{Value:} In the {\Vertical} indexing tree, each $(k,g)$-leaf node comprises a unique set of nodes. These nodes are characterised by having a $k$-coreness of exactly $g$ and a $g$-coreness of exactly $k$. By this definition, the nodes in each $(k,g)$-leaf node are exclusive to that specific coreness level, ensuring no overlap with nodes in higher coreness levels. 
    \item \textit{Link:} In the {\Vertical} indexing tree, links are categorised into two distinct types: \textit{next link} and \textit{jump link}.
    \begin{enumerate}[leftmargin=*]
          \item \underline{Next Link:} The next link, the same as in the {\Horizontal} indexing tree, connects to the subsequent leaf node in the sequence. Specifically, a $(k,g)$-leaf node is linked to the $(k+1,g)$-leaf node. This link is based on the $g$-coreness, determining the exact position of a node within the tree structure.
           \item \underline{Jump Link:} The jump link connects to the next leaf node fixing the $k$ value. For instance, a $(k,g)$-leaf node is linked to the $(k,g+1)$-leaf node. Based on the $k$-coreness, it guides the location of a node within the tree.
    \end{enumerate}
\end{itemize}

\spara{Query processing.} To process a query $Q = (k,g)$ using the indexing tree $T$, we employ a specific procedure to accurately find the $(k,g)$-core. The first step is to find the corresponding $(k,g)$-leaf node in the indexing tree. Starting from this node, the process involves iterative traversal of subsequent leaf nodes along the jump link until the terminal leaf node is reached. Subsequently, for each of these leaf nodes, traversal continues along the next link until the terminal leaf node is encountered. The $(k,g)$-core is then obtained by aggregating all nodes during these traversals across the respective leaf nodes.

\begin{figure}[t]
\centering
\includegraphics[width=0.9\linewidth]{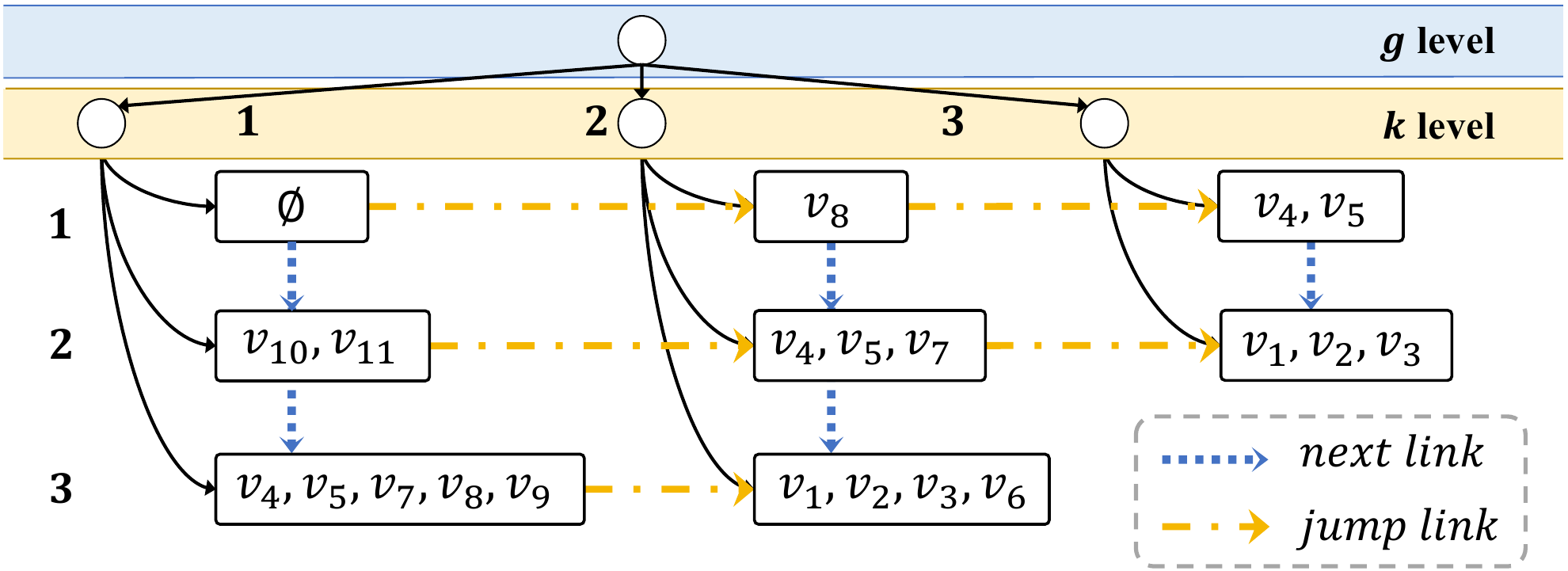}
\vspace{-0.2cm}
\caption{{\Vertical} indexing tree structure} 
\vspace{-0.5cm}
\label{fig:vertical_tree}
\end{figure}


\subsubsection{Complexity analysis}
\textcolor{black!0}{This text is completely invisible.}

\spara{Construction time complexity.} The {\Vertical} algorithm extends the {\Horizontal} approach by visiting the ($g^*-1$) children of root nodes and performing additional set difference operations between $k^*_g$ leaf nodes. 
The complexity of these operations is bounded by $O(g^* \cdot |V|)$. Therefore, the overall time complexity for constructing the {\Vertical} indexing tree is $O(g^* \cdot (P + k^* \cdot |V|))$, which is comparable to the time complexity of the {\Horizontal} indexing tree.

\spara{Query processing time complexity.} In worst case, {\Vertical} requires traversing all leaf nodes while aggregating the nodes. Hence, the time complexity for obtaining a specific $(k,g)$-core is $O(k^* \cdot g^* \cdot |V|)$.

\spara{Space complexity.} The space complexity of the {\Vertical} indexing tree is equivalent to that of the {\Horizontal} indexing tree, which is $O(g^* \cdot |V|)$. This is because, in the worst case, nodes may not be stored adjacently within the tree. For instance, if a node $v_i$ is present in the $(k,g)$-leaf node and also belongs to the $(k-1, g+1)$-core, duplication of nodes within the indexing tree is inevitable.

\begin{algorithm}[t]
\SetAlgoLined
\small
\SetKwData{break}{break}
\SetKwData{false}{false}
\SetKwData{true}{true}
\SetKwFunction{update}{updateLeafNode}
\SetKwFunction{leaf}{Leaf}
\KwIn{Hypergrpah $G=(V,E)$}
\KwOut{indexing tree $T$}

\tcc{Lines 1-12 in Algorithm~\ref{alg:horizontal_level_index_construction}}
\setcounter{AlgoLine}{12}
\For{$g \gets 1$ \KwTo $g^*-1$}{
    \For{$k \gets 1$ \KwTo $k^*_{(g+1)}$}{
    \leaf{$k$,$g-1$}.jump $\leftarrow$ \leaf{$k$,$g$}\;
    $s \leftarrow $ \leaf{$k$,$g-1$} $\setminus$ \leaf{$k$,$g$}\; 
    \update{$T$,$k$,$g-1$,$s$}\;
    }
    
}
\Return{$T$}
\caption{Vertical Locality-based indexing}
\label{alg:vertical}
\end{algorithm}

\begin{example}
Figure~\ref{fig:vertical_tree} shows the {\Vertical} indexing tree for the hypergraph $G$ from Figure~\ref{fig:graphs}(a). It demonstrates that the {\Vertical} indexing tree reduces redundancy by compressing vertical overlaps. For example, the $(3,1)$-leaf in the {\Vertical} indexing tree contains only $\{v_4, v_5, v_7, v_8, v_9\}$, as {\Vertical} removes duplicate nodes from the same $k$-level branches. To process the query $Q = (2,2)$, it first finds the $(2,2)$-leaf and uses the next links to aggregate nodes from leaves such as $(3,2)$. Then, it follows the jump links to include nodes from leaves like $(2,3)$. As a result, it retrieves nodes $\{v_1, \dots, v_7\}$, corresponding to the $(2,2)$-core. However, even after {\Vertical} compression, nodes $v_1, v_2, v_3, v_4, v_5$, and $v_7$, which are placed in diagonally positioned leaf nodes, still appear multiple times within the indexing tree.
\end{example}

\begin{figure}[h]
\centering
\includegraphics[width=0.95\linewidth]{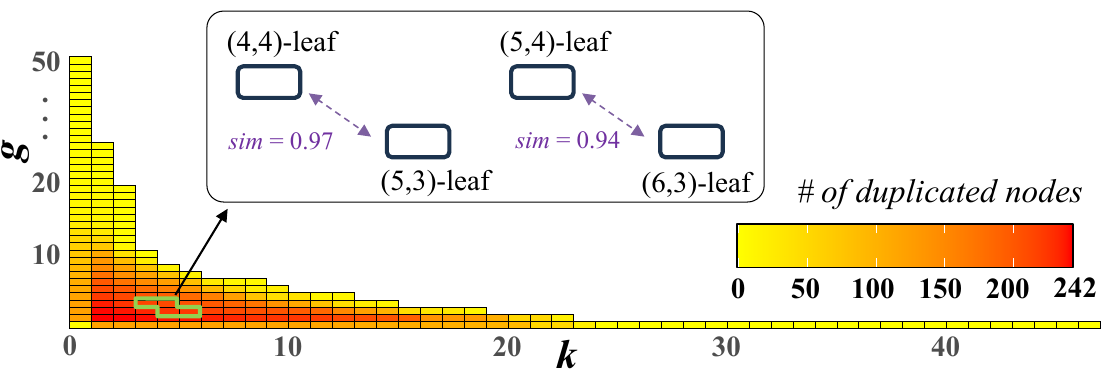}
\vspace{-0.3cm}
\caption{Number of nodes in diagonal leaf nodes} 
\vspace{-0.5cm}
\label{fig:num_of_diag}
\end{figure}

\subsection{\mbox{Diagonal Locality-based Indexing({\Diagonal})}}
In this section, we introduce the Diagonal Locality-based Indexing technique. Beyond eliminating redundancies in hierarchical $(k,g)$-cores, it captures and addresses the locality among non-hierarchical $(k,g)$-cores, particularly focusing on the relationships between diagonally placed leaf nodes.


\spara{Rationale of {\Diagonal} approach.}
The {\Diagonal} indexing technique is motivated by the observation that strong locality exists even among non-hierarchical cores in the {\Naive} indexing tree. To validate this, we analysed the \textit{Contact} dataset~\cite{arafat2023neighborhood} to measure the number of common nodes between diagonally adjacent leaf nodes (e.g.,  the $(k-1, g)$-leaf node and the $(k, g-1)$-leaf node). Utilising the Jaccard similarity, we computed the overlap between these diagonally adjacent leaf nodes and averaged scores for cases where a leaf node has two diagonally adjacent leaf nodes. 
Figure~\ref{fig:num_of_diag} presents that there is significant overlap between diagonally adjacent leaf nodes, which points out the presence of strong diagonal locality in real-world networks. This observation supports the rationale behind the {\Diagonal} approach, which aims to reduce redundancy by identifying and storing overlapping nodes in dedicated auxiliary nodes. This process extends to diagonally adjacent auxiliary nodes, ensuring that no duplicate nodes remain. To achieve this, the {\Diagonal} indexing tree is constructed on top of the {\Vertical} indexing tree, further compressing the structure by leveraging diagonal locality.

\spara{Auxiliary nodes.}
In the {\Diagonal} indexing tree, an auxiliary node plays a key role in indexing. For any given $(k,g)$-leaf node, there exists a corresponding auxiliary node, denoted as $(k,g)$-aux node if both $k>1$ and $g>1$.  
Formally, each auxiliary node is characterised by a key-value pair. The key represents the depth and the value is a set of nodes. The depth describes the hierarchical degree of common neighbours. For a $(k,g)$-aux node, if the depth $d$ is $1$, the corresponding value includes a set of nodes that are commonly appearing in both $(k-1, g)$-leaf node and $(k, g-1)$-leaf node in the {\Vertical} indexing tree. When the depth is $2$, it represents a higher level of hierarchical common neighbours. Specifically, a set of nodes at $d=2$ in a $(k,g)$-aux node indicates their presence in both $(k-1, g)$-aux and $(k, g-1)$-aux nodes at $d=1$. This implies that these nodes are common across $(k-2, g)$-leaf node, $(k-1, g-1)$-leaf node, and $(k, g-2)$-leaf node simultaneously.
This structure of auxiliary nodes, with their depth-based key-value representation, allows for a layered understanding of node relationships and reduces the node duplication within the {\Diagonal} indexing tree.

\begin{lemma}\label{lemma:co_existence_impossible}
If a node $y$ exists in a $(k,g)$-leaf node, then it cannot coexist in the corresponding $(k,g)$-aux node.
\end{lemma}

\begin{proof}
The existence of a node $y$ in a $(k,g)$-leaf node implies that the $k$-coreness of the node is $g$, and simultaneously, its $g$-coreness is $k$. However, if the same node $y$ also exists in the $(k,g)$-aux node, it implies that the $k$-coreness of $y$ is less than $g$, or its $g$-coreness is less than $k$. It is contradictory to the conditions for belonging to the $(k,g)$-leaf node. Therefore, it is not possible for a node to coexist in both a $(k,g)$-leaf node and its corresponding $(k,g)$-aux node.
\end{proof}

\begin{theorem}\label{theorem:diag}
Given a $(k,g)$-aux node, for any pair of nodes at different depths $d'$ and $d''$, there are no duplicate nodes.
\end{theorem}

\begin{proof}
Suppose that in a $(k,g)$-aux node of the {\Diagonal} indexing tree, where $k > 1$ and $g > 1$, there exist sets of nodes at depths $d'$ and $d''$, and a node $y$ appears in both sets. For simplicity, we assume $d' < d''$. First, consider the existence of node $y$ at depth $d''$. This necessitates its existence in the leaf nodes $(k-d'', g)$-leaf node, $(k-(d''-1), g-1)$-leaf, $\cdots$, $(k, g-d'')$-leaf node. Consequently, node $y$ must also be present in the auxiliary nodes $(k-d^*, g)$-aux, $(k-(d^*-1), g-1)$-aux, $\cdots$, $(k, g-d^*)$-aux where $d^* = d''-1$. Since node $y$ is in $(k,g)$-aux, it indicates that node $y$ is iteratively appeared in all the intermediate auxiliary nodes, and finally compressed in $(k,g)$-aux.  
Since node $y$ is at depth $d'$, it indicates that node $y$ appears in the auxiliary nodes $(k-d', g)$-leaf node, $(k-(d'-1), g-1)$-leaf node, $\cdots$, $(k, g-d')$-leaf node. This is contradictory since node $y$ cannot appear at $(k,g)$-leaf node and $(k,g)$-aux node simultaneously as we have checked in Lemma~\ref{lemma:co_existence_impossible}. Thus, it is proven that in a $(k,g)$-aux node, there cannot be a pair of corresponding nodes at different depths containing duplicate nodes.
\end{proof}

\begin{corollary}\label{cor:no_duplicate}
\textcolor{black}{
In any given auxiliary node $(k,g)$-aux node, a node $v$ will not be present at multiple depths.
}
\end{corollary}

\textcolor{black}{
Owing to Corollary~\ref{cor:no_duplicate}, it is assured that the incorporation of auxiliary nodes does not lead to an increase in the total number of nodes stored. Thus, this ensures that the {\Diagonal} indexing tree contains fewer nodes compared to the {\Vertical} indexing tree.
}

\subsubsection{Indexing framework}

\textcolor{black!0}{This text is completely invisible.}

\begin{table*}[t]
\caption{Summary of the algorithms}
\vspace{-0.3cm}
\label{tab:eff_test}
\centering
\small
\begin{tabular}{c||cc|cc|c|cc|c}
\hline \hline
\textbf{Algorithms}  & \multicolumn{2}{c|}{\textbf{Querying efficiency}} & \multicolumn{2}{c|}{\textbf{Space efficiency}} & {\textbf{Value of $\boldsymbol{(k,g)}$-leaf}} & \multicolumn{2}{c|}{\textbf{Links}} & {\textbf{Aux}} \\ \hline \hline
\Naive & \multicolumn{1}{c|}{$\star$$\star$$\star$$\star$$\star$} & $O(1)$ & \multicolumn{1}{c|}{$\star$} & $O(k^*\cdot g^*\cdot |V|)$ & $(k,g)$-core & \multicolumn{2}{c|}{No link} & $\times$ \\ \hline
\Horizontal & \multicolumn{1}{c|}{$\star$$\star$$\star$$\star$} & $O(k^* \cdot |V|)$ & \multicolumn{1}{c|}{$\star$$\star$$\star$} & $O(g^*\cdot |V|)$ & Nodes with  same $g$-coreness & \multicolumn{2}{c|}{Next link} & $\times$ \\ \hline
\Vertical & \multicolumn{1}{c|}{$\star$$\star$$\star$} & $O(k^*\cdot g^* \cdot |V|)$ & \multicolumn{1}{c|}{$\star$$\star$$\star$$\star$} & $O(g^*\cdot |V|)$ & Nodes with  same $k, g$-coreness & \multicolumn{1}{c|}{Next link} & Jump link & $\times$ \\ \hline
\Diagonal & \multicolumn{1}{c|}{$\star$$\star$$\star$} & $O(k^*\cdot g^* \cdot |V|)$ & \multicolumn{1}{c|}{$\star$$\star$$\star$$\star$$\star$} & $O(g^*\cdot |V|)$ & \makecell{Aux : diagonally adjacent common nodes \\ Leaf: {\Vertical} $(k,g)$-leaf $\setminus$ Aux } & \multicolumn{1}{c|}{Next link} & Jump link & $\bigcirc$ \\ \hline \hline
\end{tabular}
\vspace{-0.2cm}
\end{table*}

\spara{Indexing tree construction.} The construction process starts with the $(1,1)$-leaf node of the {\Vertical} indexing tree, progressively increasing the $g$ value. Any nodes shared by diagonally adjacent leaf nodes are stored in an auxiliary node at depth $1$, located at the intersection of these leaf nodes. If no leaf node exists at this intersection, an empty one is created, linked to the auxiliary node, and the pointers are adjusted accordingly. For example, common nodes between the $(k'+1, g')$-leaf node and the $(k', g'+1)$-leaf node are stored in the $(k'+1, g'+1)$-aux node at depth $1$. These duplicated nodes are then removed from the current leaf node. Additionally, if diagonally adjacent auxiliary nodes contain common nodes at the same depth $d'$, they are transferred to the auxiliary node at the intersection with depth $d'+1$. For instance, if a node $v$ is present in both $(k'+1, g')$-aux and $(k', g'+1)$-aux nodes at depth $d$, it is moved to the $(k'+1, g'+1)$-aux node at depth $d+1$. The duplicated nodes are then removed from the corresponding auxiliary nodes. If no further diagonally adjacent nodes are found, the duplicates are addressed via the next pointer of the current auxiliary node. This iterative process continues until the next pointer of the current leaf node reaches the maximum $k$ value.

\spara{Leaf node.} \textcolor{black}{As for its components, a leaf node in {\Diagonal} comprises a set of nodes and two types of links. A unique feature of {\Diagonal} is the presence of ``auxiliary nodes''. These nodes store common nodes found in diagonally adjacent nodes. Consequently, sets of common nodes that are diagonally adjacent are not directly stored within the leaf nodes but are instead maintained through these auxiliary nodes.}

\spara{Query processing.} 
Given a query $Q = (k,g)$, the initial step is to identify the corresponding $(k,g)$-leaf node in the tree. From this node, the process entails traversing all reachable leaf nodes via both jump and next links. Unlike in the {\Vertical} indexing tree, where all nodes are simply aggregated to obtain the result, the {\Diagonal} indexing tree requires aggregating nodes in auxiliary nodes during this traversal. It is crucial to note that auxiliary nodes located in the starting leaf node, as well as those nodes beyond the relevant depth, are not actual nodes corresponding to the query and should not be included in the aggregation process. This distinction is key to ensuring accurate query results in the {\Diagonal} indexing tree.  

\subsubsection{Complexity analysis}
\textcolor{black!0}{This text is completely invisible.}

\begin{figure}[h]
\centering
\includegraphics[width=0.95\linewidth]{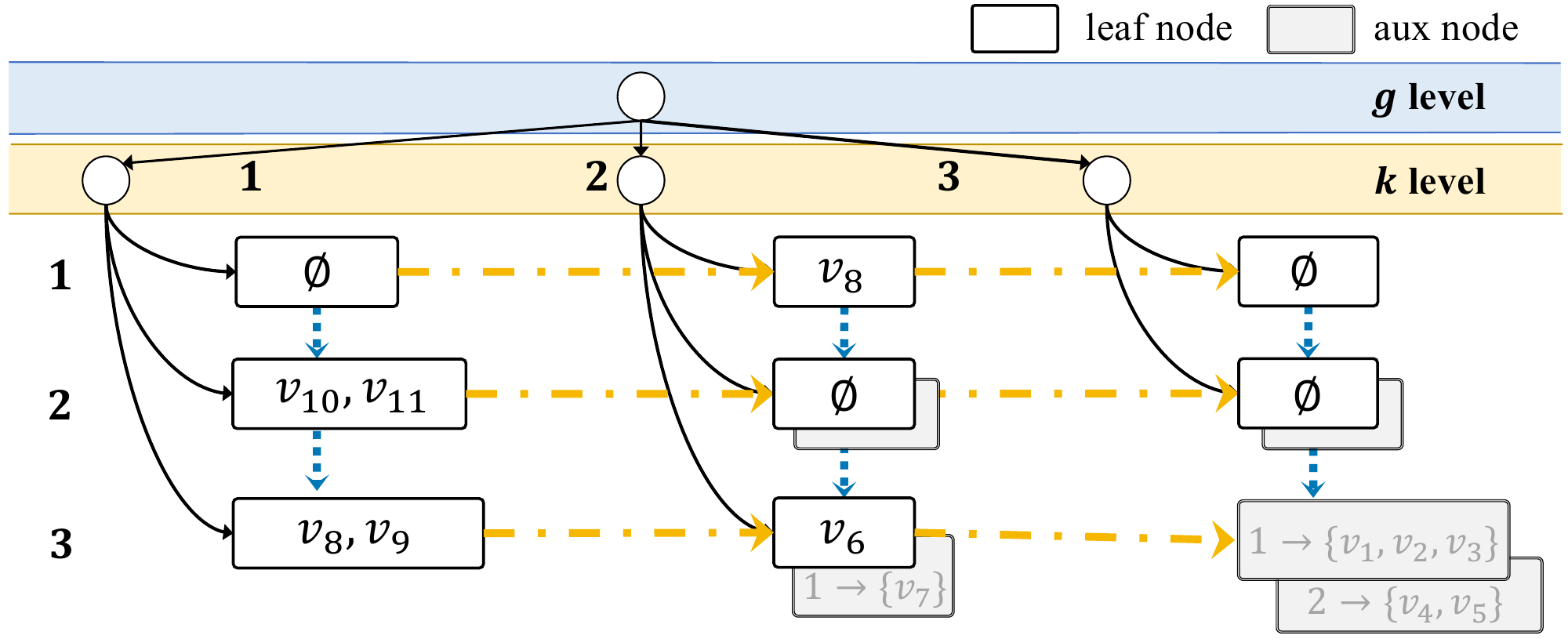}
\vspace{-0.3cm}
\caption{{\Diagonal} indexing tree structure} 
\vspace{-0.5cm}
\label{fig:diag}
\end{figure}

\spara{Construction time complexity.} 
The construction of the {\Diagonal} indexing tree builds upon the {\Vertical} indexing tree. After constructing the {\Vertical} indexing tree, adjacent branches are traversed to identify common nodes in diagonal leaf nodes, which are then stored in auxiliary nodes. 
Hence, the time complexity for constructing the {\Diagonal} indexing tree is $O(g^* \cdot (P + k^* \cdot |V|))$.

\spara{Query processing time complexity.} The time complexity for query processing in the {\Diagonal} indexing tree is comparable to that of the {\Vertical} tree. While the presence of auxiliary nodes may incur additional operations, these do not significantly impact the overall time complexity. Therefore, the time complexity for query processing in the {\Diagonal} tree is still bounded by that of the {\Vertical} indexing tree, specifically $O(k^* \cdot g^* \cdot |V|)$.

\spara{Space complexity.} The space complexity of the {\Diagonal} indexing tree is comparable to that of the {\Horizontal} indexing tree, but it offers significant advantages by effectively capturing and eliminating redundancies between diagonally adjacent nodes. This approach ensures greater space efficiency, particularly in real-world applications where diagonal locality is prevalent. 

\begin{example}
Figure~\ref{fig:diag} shows the {\Diagonal} indexing tree with auxiliary nodes based on the hypergraph in Figure~\ref{fig:graphs}(a). Almost all redundant nodes are handled using auxiliary nodes. For querying $Q=(3,2)$, we start at the $(3,2)$-leaf node, skip auxiliary nodes of the starting leaf, and move through the jump nodes to the terminal auxiliary node. After one jump link, in the $(3,3)$-aux node, we store only the nodes with depth $1$, resulting in ${v_1,v_2,v_3,v_6}$.

\end{example}

\subsection{SUMMARY OF ALGORITHMS}
Each indexing technique aims to minimise memory usage while ensuring efficient query processing. Table~\ref{tab:eff_test} analyses the algorithms' efficiency, detailing their time and space complexities, leaf node structures, link types, and auxiliary (aux) nodes. Efficiency levels are indicated using star notation ($\star$), with more stars representing greater efficiency. The table highlights the trade-off between query processing and space efficiency.


\section{EXPERIMENTS} \label{sec:experiments}



\subsection{Experiment Setup} 
We conducted extensive experiments to evaluate the practical effectiveness of our techniques and to handle key evaluation questions (\textbf{EQs}) across diverse scenarios. 

\begin{figure*}[t]
    \centering
\includegraphics[width=0.99\linewidth]{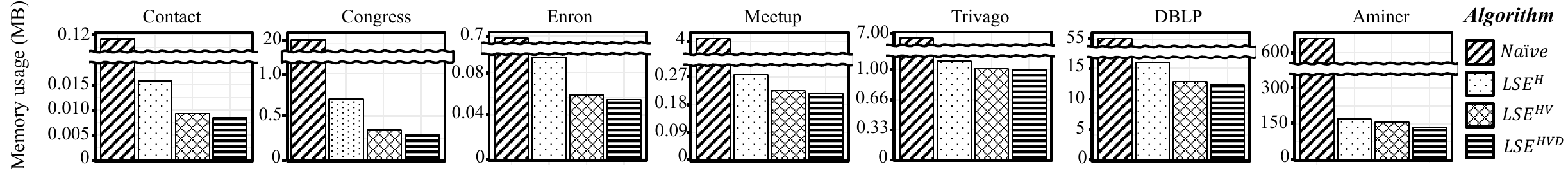}
    \vspace{-0.2cm}
    \caption{Memory usage of the indexing trees}    \vspace{-0.3cm}
    \label{fig:avg_freq}
\end{figure*}

\begin{itemize}[leftmargin=*]

\item \textbf{EQ1. Effectiveness of indexing technique}: How well do our indexing techniques avoid node duplication? 
\item \textbf{EQ2. Efficiency of index construction and query processing}: How good is the performance of index construction and query processing algorithms in terms of running time? 
\item \textbf{EQ3. Statistics of the indexing tree}: How do the statistics of indexing trees, including aspects like the \# of leaf nodes, tree depth, and the \# of nodes contained within each leaf?
\item \textbf{EQ4. Scalability}: How efficient is the query processing in benchmark hypergraphs of varying sizes? This question addresses the scalability of our techniques.
\item \textbf{EQ5. Comparative analysis with the Bi-core index model}:
This \textbf{EQ} focuses on comparing our algorithm with the other index-based approach to evaluate the efficiency of query processing.
\item {\textbf{EQ6. Case study}: 
How effective is the model in real scenarios?} 
\end{itemize}


\subsection{Experimental Setting}\label{subsection:experiment_setting}

\begin{table}[t]
\caption{Real-world dataset}
\vspace{-0.2cm}
\label{tab:data}
\centering
\small
\begin{tabular}{c||c|c|c|c|c}
\hline \hline
\textbf{Dataset}  &$\boldsymbol{|V|}$         & $\boldsymbol{|E|}$         &$\boldsymbol{\mu(N(.))}$  &$\boldsymbol{k}^{\boldsymbol{*}}$ & $\boldsymbol{g}^{\boldsymbol{*}}$\\ \hline \hline
Contact~\cite{arafat2023neighborhood}  & 242         & 12,704      & 68.74                & 47  & 54          \\\hline
Congress~\cite{benson2018simplicial} & 1,718       & 83,105      & 494.68               & 368 &1003           \\ \hline
Enron~\cite{arafat2023neighborhood}    & 4,423       & 5,734       & 25.35                & 40  & 392           \\ \hline
Meetup~\cite{arafat2023neighborhood}   & 24,115      & 11,027      & 65.27                & 121 & 250            \\ \hline
Trivago~\cite{chodrow2021hypergraph}  & 172,738     & 233,202     & 12.68                & 84  & 63          \\ \hline
DBLP\cite{arafat2023neighborhood}     & 1,836,596   & 2,170,260   & 9.05                 & 279 & 221             \\ \hline 
Aminer\cite{arafat2023neighborhood}   & 27,850,748  & 17,120,546  & 8.38                 & 610 & 730             \\ \hline \hline
\end{tabular}
\vspace{-0.5cm}
\end{table}

\spara{Real-world dataset.} 
We used seven datasets, including communication networks, co-authorship networks, and event participation networks. The statistics of the dataset are shown in Table~\ref{tab:data}, where $\mu(N(.))$ is the average neighbour size and $k^*$, $g^*$ denote the maximum $k$ and $g$ values. Note that all the datasets are publicly available. 

\spara{Algorithms.}
To the best of our knowledge, our work does not have a direct competitor in the previous literature. However, to evaluate the efficiency of indexing methods, we conducted a comparative experiment using the Bi-core index~\cite{liu2019efficient}. 


\subsection{Experimental Result}

\spara{EQ1-1. Effectiveness of indexing technique (Space usage).} 
To evaluate the impact of our indexing techniques, we measured memory usage across various real-world datasets. Figure~\ref{fig:avg_freq} shows a significant difference between the {\Naive} and {\Horizontal} indexing trees, demonstrating that using $g$-coreness substantially reduces duplication. The {\Vertical} approach also effectively reduces memory usage, and the {\Diagonal} method decreases it even further. Especially, in datasets with highly nested structures, such as the \textit{Congress} dataset~\cite{arafat2023neighborhood}, the {\Naive} indexing tree requires over $20$MB of memory. In contrast, after applying all indexing techniques, only $1.5\%$ of the memory required by the {\Naive} indexing tree is needed. It demonstrates the significant memory efficiency of our indexing techniques.


\begin{table}[h]
\centering
\vspace{-0.2cm}
\caption{Relative differences of indexing techniques}
\vspace{-0.3cm}
\label{tab:change_rate}
\small
\begin{tabular}{l||c|c|c}
\hline\hline
\textbf{Measure}    & \textbf{Contact} & \textbf{Congress} & \textbf{Enron} \\ \hline \hline
$(1-|${\Horizontal}$|$ $/$ $|${\Naive}$|)$ $\times 100$ & $86$\%   & $96$\%  & $87$\%   \\
$(1-|${\Vertical}$|$ $/$ $|${\Horizontal}$|)$ $\times 100$      &$40$\%   & $56$\%  & $38$\%   \\
$(1-|${\Diagonal}$|$ $/$ $|${\Vertical}$|)$ $\times 100$    & $9$\%   & $12$\%  & $7$\%  \\ \hline  \hline
\end{tabular}
\end{table}

\spara{EQ1-2. Effectiveness of indexing technique (Node size proportion).}
In this section, we focus on the relative differences among the {\Horizontal}, {\Vertical}, and {\Diagonal} indexing trees. To verify these differences, we checked the proportion of node sizes in each index compared to the previous index. Table~\ref{tab:change_rate} presents the ratio of change in node size for each indexing tree on three representative datasets: \textit{Contact}, \textit{Congress}, and \textit{Enron}~\cite{arafat2023neighborhood}. Note that $|I|$ refers to the number of nodes in the indexing structure $I$. Even after the dramatic reduction achieved with the {\Horizontal} approach, the {\Vertical} approach addresses around $40\%$, up to $56$\%, of duplicates within the {\Horizontal} index. Additionally, a further indexing of around $10$\% is achieved through the {\Diagonal} approach.

\begin{table*}[h]
\centering
\caption{Running time for index construction and query processing}
\vspace{-0.3cm}
\label{tab:run_time}
\small
\begin{tabular}{c||cccc||ccccc} 
\hline \hline
\multicolumn{1}{c||}{\multirow{2}{*}{\vspace{-0.1cm}\textbf{Dataset}}} & \multicolumn{4}{c||}{\textbf{Index construction time (sec.)}} & \multicolumn{5}{c}{\textbf{Query processing time (sec.) }}                              \\ \cline{2-10} 
\multicolumn{1}{c||}{}                        & \multicolumn{1}{Sc|}{\textbf{{\Naive}}}    & \multicolumn{1}{c|}{{\textbf{\Horizontal}}} & \multicolumn{1}{c|}{{\textbf{\Vertical}}} & \multicolumn{1}{c||}{{\textbf{\Diagonal}}} & \multicolumn{1}{c|}{{\textbf{\Naive}}}    & \multicolumn{1}{c|}{{\textbf{\Horizontal}}} & \multicolumn{1}{c|}{{\textbf{\Vertical}}} & \multicolumn{1}{c|}{{\textbf{\Diagonal}}} & \multicolumn{1}{c}{\textbf{Peeling algorithm}}  \\  \hline\hline
Contact                                        & \multicolumn{1}{c|}{4.2553}   & \multicolumn{1}{c|}{3.9963}     & \multicolumn{1}{c|}{3.9968}   & 3.9971                        & \multicolumn{1}{c|}{0.0003} & \multicolumn{1}{c|}{0.0018}   & \multicolumn{1}{c|}{0.0028} &  \multicolumn{1}{c|}{0.0041} & 6.1684            \\   \hline
Congress                                        & \multicolumn{1}{c|}{21,642}   & \multicolumn{1}{c|}{20,963}     & \multicolumn{1}{c|}{20,963}   &         20,964        & \multicolumn{1}{c|}{0.0005} & \multicolumn{1}{c|}{0.0116}   & \multicolumn{1}{c|}{0.0775} &  \multicolumn{1}{c|}{0.3637} &       364.87     \\   \hline
Enron                                          & \multicolumn{1}{c|}{58.902}   & \multicolumn{1}{c|}{49.437}     & \multicolumn{1}{c|}{49.440}    & 49.442                        & \multicolumn{1}{c|}{0.0005} & \multicolumn{1}{c|}{0.0013}   & \multicolumn{1}{c|}{0.0156} & \multicolumn{1}{c|}{0.0616} & 5.6496                      \\ \hline
Meetup                                         & \multicolumn{1}{c|}{541.136}  & \multicolumn{1}{c|}{481.066}    & \multicolumn{1}{c|}{481.072}  & 481.081                       & \multicolumn{1}{c|}{0.0004} & \multicolumn{1}{c|}{0.0051}   & \multicolumn{1}{c|}{0.0323} & \multicolumn{1}{c|}{0.0717} & 61.610                  \\ \hline
Trivago                                        & \multicolumn{1}{c|}{389.590}   & \multicolumn{1}{c|}{370.602}    & \multicolumn{1}{c|}{370.622}  & 370.646                       & \multicolumn{1}{c|}{0.0003} & \multicolumn{1}{c|}{0.0022}   & \multicolumn{1}{c|}{0.0254} & \multicolumn{1}{c|}{0.0531} &161.96                  \\ \hline
DBLP                                           & \multicolumn{1}{c|}{5,655.0}     & \multicolumn{1}{c|}{4,656.2}   & \multicolumn{1}{c|}{4,656.5} & 4,656.8                     & \multicolumn{1}{c|}{0.0004} & \multicolumn{1}{c|}{0.0162}   & \multicolumn{1}{c|}{0.0174} & \multicolumn{1}{c|}{0.1029} &961.97                      \\ \hline 
\textcolor{black}{Aminer}                                           & \multicolumn{1}{c|}{153,022}     & \multicolumn{1}{c|}{125,985}   & \multicolumn{1}{c|}{125,189} & 125,997                     & \multicolumn{1}{c|}{0.0004} & \multicolumn{1}{c|}{0.0418}   & \multicolumn{1}{c|}{0.1547} & \multicolumn{1}{c|}{0.1737} &         10462.1            \\ \hline \hline
\end{tabular}
\vspace{-0.2cm}
\end{table*}

\spara{EQ2. Efficiency of index construction and query processing.} 
We analysed the index construction and query processing times for each technique. As shown in Table~\ref{tab:run_time}, index construction time increase slightly from {\Horizontal} to {\Vertical}, and then to {\Diagonal}, reflecting their interdependencies. Despite fewer computational steps, the {\Naive} approach takes much longer due to the overhead of storing many nodes across leaf nodes. For query processing, we sorted all $(k,g)$-cores by size using the {\Naive} indexing tree and selected \textbf{100 queries} representing the $1^{st}$ to $100^{th}$ percentiles \and calculated the total processing time. While {\Naive} shows consistent speed as it avoids tree traversal, query times increase gradually from {\Naive} to {\Diagonal}. On average, {\Vertical} is $5$ times slower than {\Horizontal}, and {\Diagonal} is $2$ times slower than {\Vertical}. Despite being the slowest among indexing techniques, {\Diagonal} is still up to $9,600$ times faster than the peeling algorithm. These results show that our indexing takes a reasonable time while guaranteeing online query processing.

\begin{figure}[h]
    \centering
\includegraphics[width=0.9\linewidth]{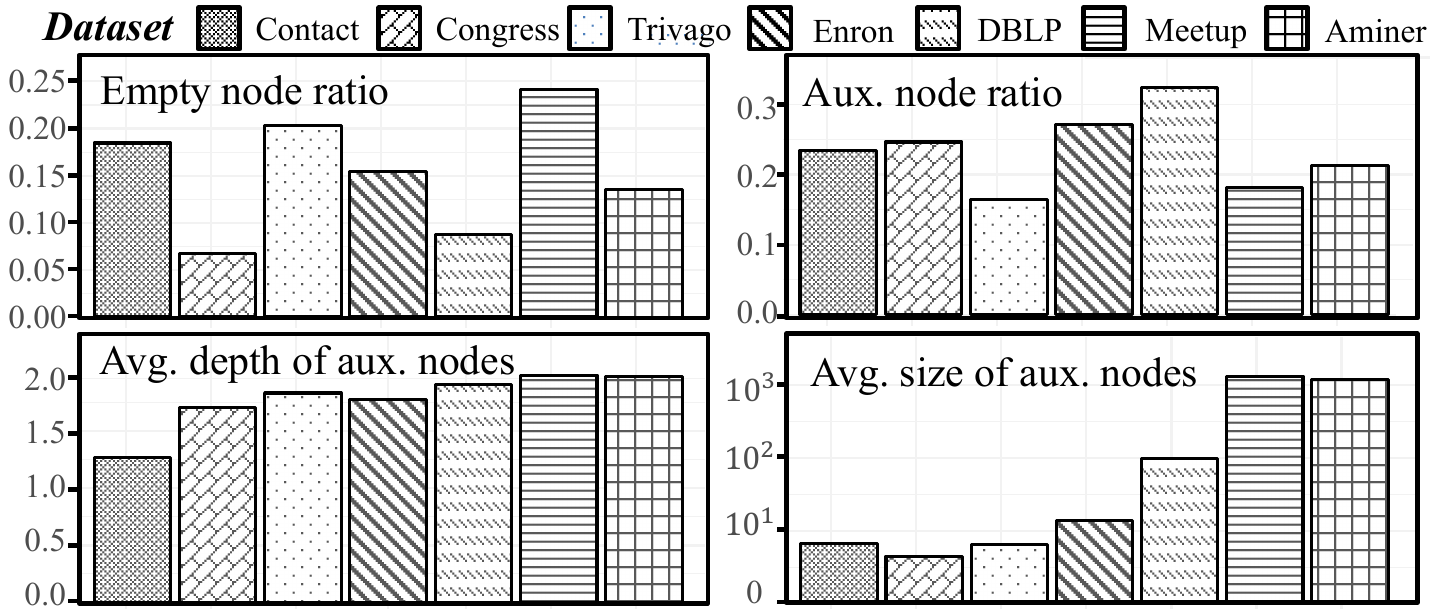}
\vspace{-0.2cm} 
    \caption{Leaf \& aux node statistics}
\vspace{-0.2cm}
    \label{fig:leaf_node}
\end{figure}

\spara{EQ3. Statistics of the indexing tree.} Figure~\ref{fig:leaf_node} presents various statistics on leaf and auxiliary (aux) nodes, offering a detailed analysis of the {\Diagonal} indexing tree structure. The proportion of empty leaf nodes to the total number of leaf nodes, which reduces duplication among linked nodes, reflects the intensity of the locality within the datasets. Additionally, the ratio of aux nodes to the total number of leaf nodes, which is close to or even exceeds $20\%$ in most datasets, indicates a high intensity of diagonal locality among many $(k,g)$-cores, pointing out the rationale of {\Diagonal}. The average depth of these auxiliary nodes mostly approaches $2$, suggesting that the auxiliary nodes handle diagonal localities consecutively over two iterations. Lastly, the average size of aux nodes tends to be proportional to the size of the datasets. These findings highlight the efficiency of the {\Diagonal} indexing tree in handling complex hypergraphs while also quantifying the locality among non-hierarchical $(k,g)$-cores, which could not be effectively captured by previous indexing techniques.

\begin{figure}[h]
\centering
\includegraphics[width=0.9\linewidth]{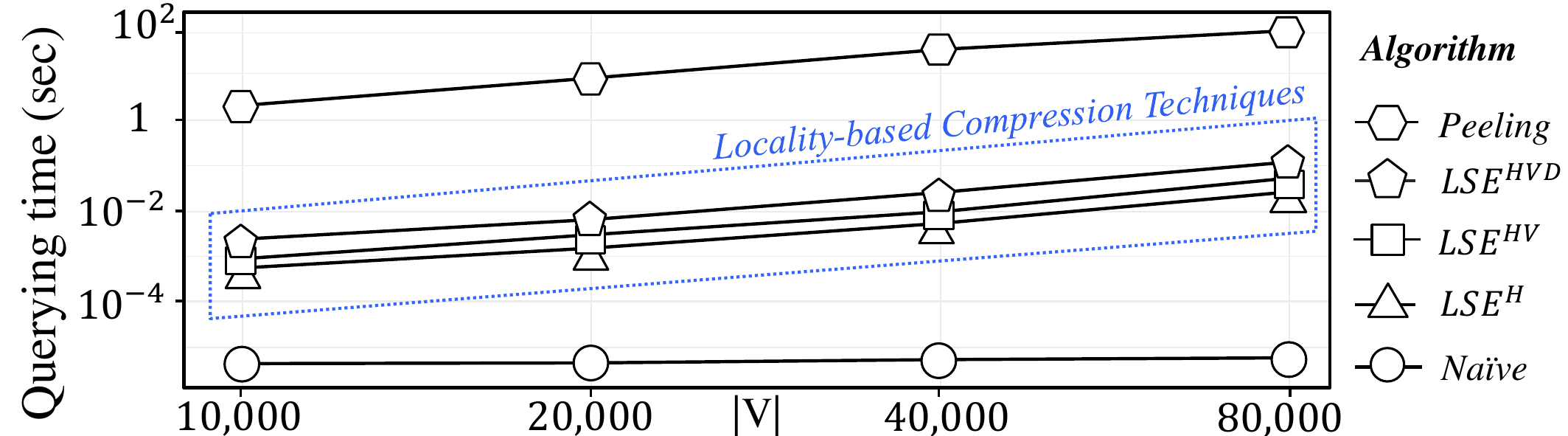}
\vspace{-0.3cm}
\caption{Scalability test (query processing time)}  
\vspace{-0.3cm}
\label{fig:syn_querying_time}
\end{figure}

\spara{EQ4. Scalability.}
Figure~\ref{fig:syn_querying_time} shows the scalability of our algorithms on synthetic hypergraphs, with node sizes of $10K$, $20K$, $40K$, and $80K$, generated using the HyperFF~\cite{ko2022growth} with default parameters. For the index-based approaches, we pre-built the indexing trees to directly compare query efficiency against the $(k,g)$-core computation algorithm. The $k$ and $g$ values were selected based on the $(k,g)$-core corresponding to $1/4$, $2/4$, and $3/4$ points in terms of the number of nodes and the average querying time for these three queries was reported. 
As illustrated in Figure~\ref{fig:syn_querying_time}, the query time for the {\Naive} indexing tree remains constant regardless of data size. We also observe that all {\lse} indexing trees demonstrate near-linear scalability. It indicates that our indexing methods are not only memory efficient but also provide reasonable query processing times.

\begin{figure}[h]
\centering
\includegraphics[width=0.9\linewidth]{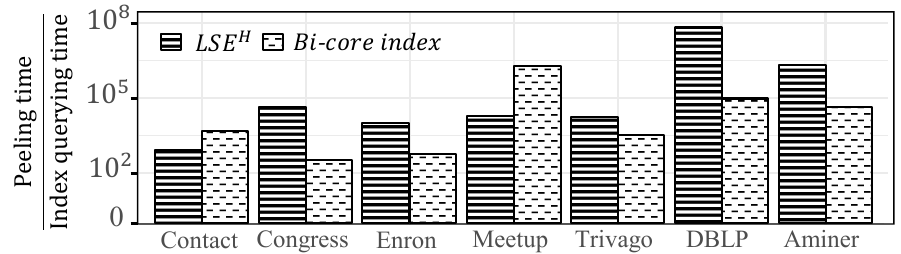}
\vspace{-0.3cm}
\caption{Comparison with Bi-core index} 
\label{fig:bicore}
\end{figure}

\spara{EQ5. Comparative analysis with the Bi-core index model.}
%
%
We compare our algorithm with the Bi-core index~\cite{liu2019efficient} to evaluate query performance against the peeling approach. The Bi-core index finds $(\alpha,\beta)$-cores in bipartite graphs, where each node on the left and right has at least $\alpha$ and $\beta$ degrees, respectively. We use the peeling algorithm from \cite{ding2017efficient} for $(\alpha,\beta)$-core identification and measure performance as the ratio of index query time to the peeling runtime. Since both the Bi-core index and our {\Horizontal} approach fix one parameter and iterate over the other, they are similarly efficient. Figure~\ref{fig:bicore} shows the execution time ratios, indicating that our algorithm performs comparably or better than the Bi-core index in most datasets, despite handling a more complex neighbour structure.

\begin{figure}[h]
\centering
\includegraphics[width=0.9\linewidth]{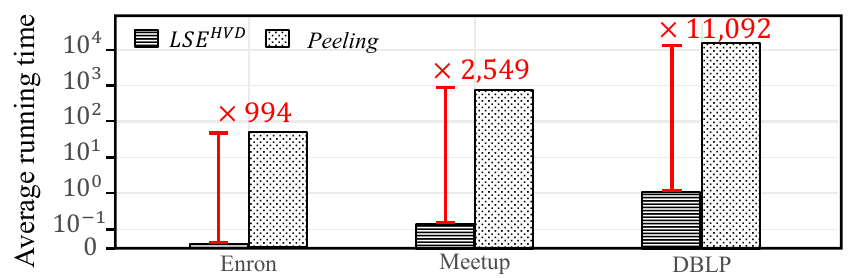}
\vspace{-0.3cm}
\caption{Case study - Size-based query processing} 
\vspace{-0.7cm}
\label{fig:case}
\end{figure}

\spara{EQ6. Case study: Size-based query processing.}
%
We explore identifying cohesive subgraphs within a specific size range without predetermined $k$ and $g$ using the \textit{Enron}, \textit{Meetup}, and \textit{DBLP} datasets, comparing methods with and without indexing. The {\Diagonal} approach is used for the index-based method. 

To simulate real-world scenarios, size lower bounds are randomly set between $30$ and $100$, with ranges $t$ between $10$ and $100$, generating $10$ queries applied to two methods:

\begin{itemize}[leftmargin=*]
\item \textbf{Without indexing tree:} $g$ is incrementally adjusted from $1$ upwards, applying the peeling strategy to find $(k,g)$ pairs satisfying size criteria. If the size drops below the lower bound, the iteration terminates, and the next $g$ is processed.
\item \textbf{With indexing tree:} For each $g$, a binary search finds indices larger than or equal to the lower bound and smaller than or equal to the upper bound, leveraging the hierarchical structure of $(k,g)$-cores. Intermediate indices are aggregated, repeating for all $g$.
%
\end{itemize}

Figure~\ref{fig:case} presents the running times of the two approaches. We observe that the index-based approach is significantly faster than the peeling algorithm. Note that to find all pairs of indices satisfying the size constraint, the algorithm must be executed multiple times. This demonstrates one of the key motivations for using an index-based approach. For the Enron dataset, we found an average of $366$ pairs, while for the Meetup and DBLP datasets, we found $615$ and $844$ pairs, respectively. These pairs enable efficient and effective querying within the given size constraints, demonstrating the practical advantages of our index-based approach in real-world applications.



\section{RELATED WORK}\label{sec:relatedwork}



\subsection{\mbox{Cohesive Subgraphs Discovery in Hypergraphs}}
Cohesive subhypergraph discovery is a key research area, aiming to identify tightly connected substructures within hypergraphs. Early models, such as the $(k, l)$-hypercore~\cite{limnios2021hcore}, define a maximal cohesive subgraph where each node has at least $k$ degree, and each hyperedge contains at least $l$ nodes, ensuring density and edge-richness. The $(k, t)$-hypercore~\cite{bu2023hypercore} improves adaptability by requiring each node to have at least $k$ degree and each hyperedge to include a proportion $t$ of nodes, offering greater flexibility. Another model, the $(k, d)$-core~\cite{arafat2023neighborhood}, focuses on maximal strongly induced subhypergraphs, composed of original hyperedges entirely within given nodes~\cite{bahmanian2015connection, dewar2017subhypergraphs}. Each node has at least $k$ neighbours and connects to $d$ edges. Recently, the $(k, g)$-core considers both a node's neighbours and their co-occurrences, potentially identifying denser subgraphs and providing a more comprehensive understanding of core structures in hypergraphs.


\subsection{Index-based Cohesive Subgraph Discovery} 
In this section, we discuss index-based cohesive subgraph models. While extensively studied in various networks, such models have yet to be explored in hypergraphs. For the $(k, p)$-core~\cite{zhang2020exploring} in simple networks, array-based indexing structures efficiently retrieve $(k, p)$-core queries. In directed graphs, tabular indices are used to find the $(k, l)$-core~\cite{fang2018effective}. The $(\alpha, \beta)$-core~\cite{ding2017efficient}, specific to bipartite graphs, is effectively managed using the Bi-core index by Liu et al.\cite{liu2019efficient}, which employs tree-based indexing for results and query processing. In multi-layered graphs, a $k$-core index using lattice and tree structures has been introduced\cite{liu2024fast}, enabling efficient query processing with parallel core computing. These advancements highlight efforts to enhance cohesive subgraph discovery and query efficiency. 
While hypergraphs can be modelled as simple or bipartite graphs, existing models struggle to fully capture their neighbour structure, necessitating a new approach to handle hypergraphs effectively.

\section{CONCLUSION} \label{sec:conclusion}

This work presents efficient indexing structures for the $(k,g)$-core in hypergraphs, enabling dynamic adjustment of $k$ and $g$ values. We propose the Nai\"ve indexing approach and the Locality-based Space Efficient indexing approach, along with three compression techniques to address space complexity. Experiments on real-world and synthetic hypergraphs confirm the efficiency of our methods. Future directions include extending our approach to dynamic hypergraphs to support efficient updates and parameter flexibility as hypergraphs evolve.

\bibliographystyle{ACM-Reference-Format}
\bibliography{sample-base}


\end{document}